\documentclass[12pt]{article}
\usepackage{amsfonts}

\usepackage{amssymb}
\usepackage{graphicx}
\usepackage{psfrag}
\usepackage{amsxtra}
\usepackage{amsmath}
\usepackage{amsthm}
\usepackage{color}
\usepackage{algorithm}
\usepackage{algorithmic_special_version}
\usepackage{subfigure}
\usepackage{setspace}

\newtheorem{theorem}{Theorem}[section]

\newtheorem{lemma}[theorem]{Lemma}

\newtheorem{claim}[theorem]{Claim}

\newtheorem{definition}[theorem]{Definition}
\newtheorem{example}{Example}[section]

\newcommand{\transposed}{{\mathrm{T}}}

\definecolor{hehngrey}{rgb}{0.5, 0.5, 0.5}

\def\ve#1{{\mathchoice{\mbox{\boldmath$\displaystyle #1$}}%
              {\mbox{\boldmath$\textstyle #1$}}%
              {\mbox{\boldmath$\scriptstyle #1$}}%
              {\mbox{\boldmath$\scriptscriptstyle #1$}}}}

\newcommand{\sgn}{{\mathrm{sgn}}}
\newcommand{\cog}{\mathrm{cog}\,}

\newcommand{\Aut}{{\mathrm{Aut}}}
\newcommand{\mymod}{\,\mathrm{mod}\,}

\newcommand{\Prob}{{\mathrm{Pr}}}

\newcommand{\LMS}{{\mathrm{LMS}}}
\newcommand{\argmax}[1]{\!\! \begin{array}[t]{c}\displaystyle{\rm argmax}\\[-2mm]\scriptstyle #1\end{array} \!}
\newcommand{\stopsets}[1]{\left|{\mathcal{S}}_{\sigma}\left(\ve{H}_{#1}\right)\right|}
\newcommand{\argmin}[1]{\!\! \begin{array}[t]{c}\displaystyle{\rm argmin}\\[-2mm]\scriptstyle #1\end{array} \!}
\newcommand{\randfun}[1]{\!\! \begin{array}[t]{c}\displaystyle{\rm rand}\\[-2mm]\scriptstyle #1\end{array} \!}

\newcommand{\mapping}{{\mathrm{map}}}

\newcommand{\Expect}{\mathrm{E}}
\newcommand{\BER}{\mathrm{BER}}
\newcommand{\FER}{\mathrm{FER}}

\newcommand{\ext}{\mathrm{e}}
\newcommand{\Eb}{E_{\mathrm{b}}}

\def\setN{\mathbb{N}}

\begin{document}

\title{{Multiple-Bases Belief-Propagation Decoding\\ of High-Density Cyclic Codes}
\thanks{Parts of the results were presented at the IEEE International
  Symposium on Information Theory (ISIT), Nice, France, 2007. The work was supported in part
  by the NSF grant CCF 0644427 awarded to Olgica Milenkovic.}}
\author{Thorsten Hehn$^{\ddagger}$, Johannes B.\ Huber$^{\ddagger}$, Olgica Milenkovic$^{\dagger}$,
Stefan Laendner$^{\ddagger}$\\$\,\!^{\ddagger}\,$Institute for Information Transmission (LIT)\\FAU Erlangen-Nuremberg,
Germany\\$\,\!^{\dagger}\,$Department of Electrical and Computer Engineering\\ University of Illinois at Urbana-Champaign, USA}

\date{\today}
\maketitle

\begin{abstract}
  We introduce a new method for decoding short and moderate length
  linear block codes with dense parity-check matrix representations of
  cyclic form, termed multiple-bases belief-propagation (MBBP). The proposed
  iterative scheme makes use of the fact that a code has many
  structurally diverse parity-check matrices, capable of detecting
  different error patterns. We show that this inherent code property
  leads to decoding algorithms with significantly better performance
  when compared to standard BP decoding. Furthermore, we describe how to
  choose sets of parity-check matrices of cyclic form amenable for multiple-bases
  decoding, based on analytical studies performed for the binary
  erasure channel. For several cyclic and extended cyclic codes, the MBBP decoding
  performance can be shown to closely follow that of
  maximum-likelihood decoders.
\end{abstract}

\textbf{Index Terms}: Algebraic Codes, Belief Propagation, Multiple-Bases Belief-Propagation Decoding, Stopping Sets.
\onehalfspacing

\section{Introduction}

Classical algebraic codes of short block length have large minimum
distance and efficient soft-decision decoding algorithms
\cite{hanetal93, fossorieretal95, lucasetal98}. Consequently, these codes
represent a good choice for low-delay applications where high
transmission reliability is required.
Algebraic codes are also frequently used as components of product
codes and parts of concatenated coding schemes. In modern storage and
communication systems, low-density parity-check (LDPC) codes are most
often used as the inner coding scheme. For this reason it is desirable
to implement a soft-input soft-output decoder for algebraic codes as a
belief-propagation (BP) algorithm. This is a reasonable choice as the
decoder can handle both types of codes. Since algebraic block codes
have dense parity-check matrices with a large number of short
cycles~\cite{halfordetal06,halfordetal06b}, BP decoders offer poor
error-correcting performance for such codes.

The use of redundant parity-check matrices for iterative decoding schemes when signaling over the
binary erasure channel (BEC) has been excessively studied. Several
authors proposed using a high number of redundant checks~\cite{schwartzetal06,
  hanetal07, ghaffaretal07, hollmannetal07} to improve the decoding
performance. This type of decoding has also drawn the attention of
researchers who studied this concept in the context of the AWGN
channel. Other authors proposed {\emph{adaptive BP
algorithms}}~\cite{kothiyaletal05,jiangetal06}, which iteratively
optimize the matrix representation during the decoding process.
Such schemes have large implementation complexity due to the
required matrix reduction after each iteration. The {\emph{random
redundant decoding}} (RRD)~\cite{halfordetal06} algorithm does not
require this type of processing and obtains very promising
results. This is accomplished by serially altering the applied
matrix representation within the decoding process. Another closely
related approach was described in \cite{andrewsetal02}, where a
simple simulation-based study using randomly chosen parity-check
matrices of the $[24,12,8]$ extended Golay code was performed.
This approach offers poor performance when compared both to the
performance of the RRD algorithm and the algorithms described in
this work.

The approach followed in this paper draws upon the prior work of
the authors on BEC decoding \cite{hehnetal08b} and introduces a
novel decoding method that operates in parallel and iterative
fashion on a collection of parity-check matrices. Using this set
of decoder representations, the algorithm performs joint output
processing in order to estimate the transmitted codeword. This
output processing can occur at various stages of decoding and it
may have various degrees of complexity. We distinguish between
techniques where the BP algorithms run separately and compare them
to schemes where the decoders are allowed to exchange information.
In the latter case, we investigate processing of the form of
\emph{metric and complexity selection}, \emph{averaging of
probabilities}~\cite{laendneretal05}, \emph{information
combining}~\cite{landetal06}, as well as certain reliability-based
schemes.

As the different representations of the parity-check matrix form
bases of the dual code, we refer to the proposed approaches as
{\emph{Multiple-Bases Belief-Propagation}} (MBBP) decoding. For
the purpose of MBBP decoding, one needs to identify classes of
parity-check matrices that \emph{jointly} offer good decoding
performance. In order to accomplish this task, we propose using
parity-check matrix design techniques originally developed for the
BEC but subsequently used for signaling over the AWGN channel.
This approach leads to good performance results, as observed both
by the authors and other researchers~\cite{zengetal05}. Moreover,
this method relies on the fact that the performance of a
parity-check matrix for the BEC is completely characterized by
combinatorial entities termed stopping sets~\cite{dietal02}; and,
that pseudocodewords for linear programming
decoders of several classes of channels represent stopping sets for any channel
in the Tanner graph~\cite{feldman03,kelleyetal07}.
Although we focus our attention on parity-check matrix construction
techniques for cyclic codes, the described concepts can be generalized
for other classes of codes as well.

The main differences between the existing RRD algorithm and the proposed MBBP scheme are that RRD
operates in a \emph{serial} fashion in terms of periodically
permuting the received word, while MBBP works in a parallel manner. Further, the RRD algorithm uses
message scaling processing between iterations that tends to increase
the overall complexity of the scheme. Contrary, the MBBP algorithm deploys the standard update rules defined by message passing decoding within its parallel cores. Finally, the RRD algorithm uses a greedy search over the Tanner graphs to
find a representation which offers good performance as well as random shuffling techniques of the variable nodes. The MBBP algorithm relies on specially designed parity-check matrix families for a given code.

The paper is organized as follows. Section \ref{sec:definitions}
introduces relevant definitions and terminology, while Section
\ref{sec:bases_selection} contains a description of the bases
selection process. Section \ref{sec:mbbp_and_variations} presents a set of different variations of MBBP decoding algorithms. Simulation results
are presented in Section \ref{sec:results}.

\section{Definitions and Terminology}
\label{sec:definitions}

We start by introducing the terminology related to stopping sets
and the BEC. We also provide the
terminology needed for describing the MBBP decoding approach.

\begin{definition}
  Let $\ve{A}$ be an $m\times n$ matrix, and let the columns of $\ve{A}$
  be indexed by the set of integers ${\mathcal{J}}=\{{0,\ldots,n-1\}}$. For a
  set ${\mathcal{I}} \subseteq {\mathcal{J}}$, we define the
  restriction of $\ve{A}$ to ${\mathcal{I}}$ as the $m \times
  |{\mathcal{I}}|$ array of elements composed of the columns of
  $\ve{H}$ indexed by ${\mathcal{I}}$.
\end{definition}

When transmitting over the BEC, stopping sets completely determine the failure
modes of iterative decoders. For completeness we define these sets below~\cite{dietal02}.

\begin{definition} For a given parity-check matrix $\ve{H}$ of an $[n,k,d]$ binary linear code $\mathcal{C}$,
a stopping set $\mathcal{S}(\ve{H})$ of size $\sigma$ is
a set ${\mathcal{I}}$ of $\sigma$ positions in the codeword for which the restriction of $\ve{H}$ to ${\mathcal{I}}$ does not contain rows of Hamming weight one.
\end{definition}

Clearly, a codeword is a stopping set and consequently the size of the
smallest stopping set of any given parity-check matrix cannot exceed
$d$.


In order to compare different parity-check matrix representations
with respect to their decoding performance, we restrict our
attention to a simple evaluation criteria: the number of stopping
sets of size less than or equal to $\sigma$, for some predefined value
$1\leq\sigma\leq d$. Given a parity-check matrix $\ve{H}$, the
number of its stopping sets of size $\sigma$ will henceforth be
denoted by
$\left|{\mathcal{S}}_{\sigma}\left(\ve{H}\right)\right|$. Although
stopping sets are known to completely characterize the performance
of iterative decoders \emph{only} for the BEC, they also represent
a partial performance indicator for transmission over the AWGN
channel. This is due to the intimate connection between
stopping sets and pseudocodewords \cite{feldman03, koetteretal03}.

As we focus our attention on codes with parity-check matrices of
cyclic form, a code category that includes cyclic codes, we also
provide the following definitions.

\begin{definition}
  Let ${\mathcal{C}}$ be a binary, linear code and
  ${\mathcal{C}}^{\perp}$ its dual. A parity-check matrix of
  ${\mathcal{C}}$ is said to be of cyclic form if it consists of
  $n-k\leq m\leq n$ consecutive cyclic shifts of one chosen codeword
  of ${\mathcal{C}}^{\perp}$. In this context, the qualifier
  ``consecutive'' implies that the $(i+1)$-th row of the parity-check matrix, $1\leq i\leq m-1$, is the cyclic
  right shift of the $i$-th row by one position.

A code ${\mathcal{C}}$ is called cyclic if any cyclic shift of a
codeword $\ve{c}\in\mathcal{C}$ is also a codeword, and it necessarily
has at least one parity-check matrix of cyclic form.
\end{definition}

For a code with at least one parity-check matrix of cyclic form,
we introduce the notion of a partition of the set of codewords of
${\mathcal{C}}^{\perp}$ and the notion of a cyclic orbit generator
({\emph{cog}}).

\begin{definition}
Let ${\mathcal{C}}$ be a binary, linear, cyclic code. Partition
the set of codewords of ${\mathcal{C}}^{\perp}$ into disjoint
orbits (subsets) consisting of cyclic shifts of one codeword. Let
one designated codeword in the orbit be the representative of the
subset. This codeword is referred to as the cyclic orbit generator
(cog).
\end{definition}

Throughout the paper we focus our attention on cogs of minimum
Hamming weight. Little technical modifications are required in the
above definition to encompass parity-check matrices that are of
cyclic form when restricted to a proper subset of columns,
e.g.\ extended cyclic codes.

Let ${\mathcal{G}}$ be the set of cyclic orbit generators with
Hamming weight equal to the minimum distance of the dual code,
$d^{\perp}$. A cyclic orbit generator
$\cog_{\ell}\in{\mathcal{G}}$,
$\ell=1,\dots,\left|{\mathcal{G}}\right|$, is used to construct a
parity-check matrix $\ve{H}_{\ell}$,
$\ell=1,\dots,\left|{\mathcal{G}}\right|$, of size $m\times n$,
$n-k \leq m\leq n$, such that the row-rank of the matrix is
$n-k$\footnote{Here, and throughout the paper, we only consider
cogs that generate parity-check matrices with row-rank $n-k$. As a
consequence, we use the word \emph{bases} to describe the
underlying matrices, although the considered structures are
actually \emph{frames}. Frames are over-complete systems in which
every element of a vector space can be represented in a not
necessarily unique manner \cite{kovacevicetal07}.}. This matrix
consists of $m$ consecutive right shifts of $\cog_{\ell}$. To
avoid identical rows in $\ve{H}_{\ell}$ even if $m=n$ holds, only cogs with a period
of $n$, i.e. cogs with a cyclic orbit that consists of $n$
distinct shifts, are considered for the construction process.

Note that a redundant parity-check matrix of cyclic form can have
at most $n$ distinct rows\footnote{Cyclic matrices with $m=n$ rows
are
  also referred to as circulant matrices.}. The larger the value of
$m$, the larger the hardware implementation complexity of the BP
decoder. Nevertheless, based on extensive computer simulations, it
was observed that for decoding of algebraic codes signaled over
the AWGN channel the best decoding performance is achieved for
$m=n$. This finding holds for both the {\emph{bit error
    rate}} ($\BER$) and {\emph{frame error rate}} ($\FER$).
The reason supporting this observation is intuitively clear.
Consider a parity-check matrix of cyclic form for which $m=n-k$,
as shown in (\ref{eq:cyclic_parity_check_matrix}) for $m=3,
n=7$. Here, the symbol $x$ serves as placeholder for the bits of
the generating cog of the matrix.
\begin{equation}
\left(\begin{array}{cc|ccc|cc}1&x&x&x&1&0&0\\0&1&x&x&x&1&0\\0&0&1&x&x&x&1\end{array}\right)
\label{eq:cyclic_parity_check_matrix}
\end{equation}

As can be seen from (\ref{eq:cyclic_parity_check_matrix}), not all of
the seven bits participate in the same number of parity-check
equations - the column degrees of the parity-check matrix vary with the column.
There exist at least two bits (including the first and last) that
participate in only one parity-check, and therefore have very low
probability of being correctly decoded.
Depending on the particular choice of the cog, the set of symbols
at the beginning and at the end of the codeword is strongly
restricted with respect to the maximum number of parity-checks it
can participate in.
%
%
This problem can be solved by setting $m=n$: such a row-redundancy
allows for achieving equal error protection for all code
symbols~\cite{linetal04}. Therefore, square parity-check matrices
will be used throughout the remainder of this paper.

We conclude this section by introducing the notion of a
cog {\emph{family}} of a set of parity-check matrices.

\begin{definition}
Let ${\mathcal{F}}_1$, ${\mathcal{F}}_2$, $\dots$, ${\mathcal{F}}_F$ be a partition of the set of indices $\left\{1,\dots,\left|{\mathcal{G}}\right|\right\}$, i.e.

\[
{\mathcal{F}}_1\cup{\mathcal{F}}_2\cup\dots\cup
{\mathcal{F}}_F=\{1,\dots,\left|{\mathcal{G}}\right|\}\mbox{ and }
{\mathcal{F}}_f \cap {\mathcal{F}}_{f'}=\emptyset,\, \forall \;
f\not= f',
\]
so that the ``stopping set performance'' of $\ve{H}_{\ell}$ is
fixed within each family ${\mathcal{F}}_f$, for all
$\ell\in{\mathcal{F}}_f,$ and for all $f\in\{1,\dots,F\}$, and
this claim is true for all families in the partition. The
``stopping set performance'' of a parity-check matrix, defined for
both the BEC and AWGN channel, is the number of stopping sets of
size up to and including $d$. We refer to the set
$\{\cog_{\ell}\}$, ${\ell}\in{\mathcal{F}}_f$, as the $f$-th cog
family. \label{def:family}
\end{definition}

\section{Bases Selection for MBBP decoding}
\label{sec:bases_selection}

Recall that a linear $[n,k,d]$ code ${\mathcal{C}}$ is uniquely
defined by a parity-check matrix $\ve{H}$ of rank $n-k$ or a
generator matrix $\ve{G}$ of rank $k$. There usually exists a
large number of generator and parity-check matrices for the same
code. For BP decoding over AWGN channels, one usually seeks a
sparse parity-check matrix $\ve{H}$. The performance of the
algorithm strongly depends on some additional structural
properties of $\ve{H}$, such as the number and weight of
pseudocodewords.

Adding redundant rows to parity-check matrices improves the
performance of iterative decoding for the BEC, but usually has
detrimental effects on BP decoders correcting data signaled over
the AWGN channel. This can be attributed to the increase of the
number of short cycles and the density of the matrix. But, as was
shown by the authors in~\cite{laendneretal06}, adding judiciously
chosen redundant rows may increase the minimum weight of
pseudocodewords (and trapping sets) of the given parity-check
matrix.

In order to exploit the benefits offered by redundant parity-check
matrices with respect to pseudocodeword performance, while
controlling the negative effects on the cycle lengths, the
following approach can be used. Rather than decoding a received
word in terms of only one parity-check matrix, one can use a
collection of parity-check matrices, each with small
row-redundancy, in parallel. The results of the decoders operating
on different parity-check matrices can then be appropriately
combined. This is the main idea behind MBBP decoding, and for this
purpose, we propose to develop good heuristic techniques for
identifying parity-check matrices which offer both good individual
and joint decoding performance.

For cyclic algebraic codes, we proved in a companion paper
\cite{hehnetal08b} that parity-check matrices that consist of cyclic
shifts of carefully chosen cogs offer excellent stopping set
properties \cite{hehnetal06}. In what follows, we focus on identifying
families of cogs that obtain very good decoding performance for the
AWGN channel. As already pointed out, the iterative decoding
performance of a fixed parity-check matrix used over the AWGN channel
is, to a certain extent, strongly correlated with its BEC performance,
so that we use the total number of stopping sets of size up to $d$ as
our cog family optimization criteria.

Since cogs from the same family define parity-check matrices of cyclic
form with identical properties for the BEC and comparable properties
for the AWGN channel, it is desirable to identify the family with the
best performance and then exclusively use cogs from this family to
build matrices for MBBP decoding. Identifying all families along with its members by counting stopping sets in the corresponding matrices is computationally expensive \cite{mcgregoretal07}. Also, storing all cogs used for decoding
can be prohibitively costly, especially for long codes and MBBP
decoders with many bases. In order to avoid these problems, we propose
to use a cog \emph{mapping technique} that relies on the notion of the
\emph{automorphism group} of a code.
\begin{definition}\cite[Ch.\ 8]{macwilliamsetal77}
The permutations which send $\mathcal{C}$ into itself, i.e.\
codewords go into (possibly different) codewords, form the
automorphism group of the code $\mathcal{C}$, denoted by
$\Aut(\mathcal{C})$. If $\mathcal{C}$ is a linear code and
$\mathcal{C}^{\perp}$ is its dual code, then
$\Aut(\mathcal{C})=\Aut(\mathcal{C^{\perp}})$.
\end{definition}

It is straightforward to prove that there exists a set of
permutations ${\mathcal{P}}$ in the automorphism group of a cyclic
code which map one cog into another cog from the same family.
Fixing one cog, and then applying a subset of permutations from
${\mathcal{P}}$ to it, represents a convenient way for generating
redundant parity-check matrices with identical densities and
comparable performance under MBBP decoding.

We provide next a partial characterization of the set
$\mathcal{P}$ for cyclic codes. More precisely, we describe how to
find a large set of permutations $\mathcal{P}$ for which
$\stopsets{a}=\stopsets{b}$, for $\sigma\leq n$, where the
generating cogs of $\ve{H}_a$ and $\ve{H}_b$ satisfy
$\cog_{b}=\kappa(\cog_{a})$, and where
$\kappa(\cdot)\in{\mathcal{P}}$. Here, $\kappa(\cog_{a})$ is used
to denote the action of the permutation $\kappa$ on the vector
$\cog_{a}$. Note that the same framework can be used when only
stopping sets up to a size smaller than or equal to $d$ are
considered.

It is well known that the automorphism group of a cyclic code
contains two classes of permutations~\cite{macwilliamsetal77}:

\begin{itemize}
\item[$P_1$:]{The class of cyclic permutations
$\alpha^0,\alpha^1,\ldots,\alpha^{n-1}$, where

\begin{eqnarray*}
\alpha: i &\to& (i+1) \, \mymod n,\; i=0,\ldots,n-1, \\
\alpha(\ve{c})&=&(c_{n-1},c_0,c_1,\ldots,c_{n-2}).
\end{eqnarray*}

} \item[$P_2$:]{The class of permutations
$\beta^0,\beta^1,\ldots,\beta^{h-1}$, where

\begin{eqnarray}
\label{eq:perm_i_to_2i}
\beta: i &\to& (2\cdot i) \, \mymod n,\; i=0,\ldots,n-1,\\
\beta( \ve{c})&=&(c_0\,c_{(n+1)/2}\,c_1\,\ldots c_{(n-1)/2}),\nonumber
\end{eqnarray}
and where $h$ denotes the cardinality of the cyclotomic coset of
the $n$-th roots of unity that contains one. Above, all subscripts
are taken modulo $n$. For extended cyclic codes, the described
permutations are only to be applied to the cyclic part of the
codeword.}
\end{itemize}

We find the following definition useful for our subsequent
derivations.

\begin{definition}
Let $\kappa$ be a permutation operating on a vector $\ve{v}$,
resulting in a vector $\ve{t}=\kappa(\ve{v})$. We define the
$\kappa$-permutation of an $m\times n$ matrix $\ve{V}$ as a
matrix obtained by permuting each row of $\ve{V}$ according to
$\kappa$. In this setting, $\ve{T}=\kappa(\ve{V})$ is used to
denote $\ve{T}(\mu,:)=\kappa(\ve{V}(\mu,:))$, $\mu=1,\dots,m$,
where $\ve{X}(\mu,:)$ represents the $\mu$-th row of the matrix
$\ve{X}$.
\end{definition}

The following straightforward results provide a partial
characterization of the set of permutations $\mathcal{P}$ of
cyclic codes. All proofs rely on the fact that $\alpha^j\, \theta
=\theta \, \alpha^{j'}$, for any integer $j$ and some integer
$j'$, and for $\theta \in \mathcal{A}(n)$, the affine group of
order $n$.

\begin{lemma} A necessary and sufficient condition for
\begin{eqnarray*}
\theta(\alpha^{j}(\cog_{a}))&=&\alpha^{j'}(\theta(\cog_{a})),\\
&&\; \text{for all} \; j\in\{0,\dots,n-1\},\\ &&\;\text{and
some}\; j'\in\{0,\dots,n-1\},
\end{eqnarray*}
to hold is that $\theta(\cdot )$ is an affine permutation, defined
as
\begin{eqnarray*}
\theta:&& i \to q \cdot i + \omega\, \mymod n,\\
&&\text{for some} \;\, q\in\{0,\dots,n-1\},\, \omega
\in\{0,\dots,n-1\}, \text{and}\; i=0,\dots, n-1,
\end{eqnarray*}
such that ${\mathrm{gcd}}(q,n)=1$. Here, ${\mathrm{gcd}}(q,n)$ is used to denote the
greatest common divisor of $q$ and $n$.
\label{lemma:affine_permutation}
\end{lemma}

\begin{proof}
The claim of Lemma \ref{lemma:affine_permutation} can be rewritten as
\begin{eqnarray*}
\theta(i+j)&=&\theta(i)+j'\mymod n\\&& j\in\{0,\dots,n-1\},\,
j'\in\{0,\dots,n-1\},\,i=0,\dots,n-1,
\end{eqnarray*}
where $\theta(i)$ denotes the action of the permutation $\theta$
on the coordinate $i$.

The former equality is true if and only if $\theta(i)$ is a
linear function of the form $q\cdot i +\omega$ for which ${\mathrm{gcd}}(q,n)=1$.
\end{proof}
The lemma asserts that cyclic permutations commute (up to a cyclic
shift) with all affine permutations in a symmetric group.

\begin{example}
  Consider $n=5$, $j=1$, and $\theta:\, i \to i+2 \mymod\,n,\;
  i=0,\ldots, n-1$.  Then,
  $\theta(\alpha(10100))=\alpha(\theta(10100))=10010$ holds. In this
  special case, both $j$ and $j'$ are equal to one.
\end{example}

\begin{claim} If $\cog_{b}=\alpha^{j}(\cog_{a})$, for some $j\in\{0,\dots,
n-1\}$, then
\[
\stopsets{a}=\stopsets{b},\, \; \text{for all}\; \, \sigma\leq n.
\]
\end{claim}

\begin{proof}
Applying $\alpha^j$ to $\cog_{a}$ cyclically permutes the rows of
$\ve{H}_{a}$. This cyclic permutation preserves all stopping sets,
which proves the claimed result.
\end{proof}
\begin{claim} If $\cog_{c}=\beta^{j}(\cog_{a})$, for some $j\in\{0,\dots,h-1\}$, then
\[
\stopsets{a}=\stopsets{c}, \, \; \text{for all}\; \, \sigma\leq n.
\]
\label{claim:permutation_beta}
\end{claim}
\begin{proof}
It is straightforward to see that
\begin{eqnarray*}
\ve{H}_{c}(\mu,:)&=&\alpha^{\mu-1}(\ve{H}_{c}(1,:))=\alpha^{\mu-1}
(\beta^{j}(\ve{H}_{a}(1,:)))\\&=&\beta^j(\alpha^{\mu'-1}(\ve{H}_{a}(1,:))),
\end{eqnarray*}
since $\beta^j$ is an affine
permutation with $q=2$ and $\omega=0$. As a result, $\ve{H}_{c}$
can be transformed into $\ve{H}_{a}$ through row- and
column-permutations.
\end{proof}
%
%

We conclude that $\mathcal{A}(n) \cap \Aut(\mathcal{C}) \subseteq
\mathcal{P}$: in other words, applying affine transforms from the
automorphism group of the code to one chosen cog in the family
produces cogs that generate parity-check matrices with identical
stopping set distributions. Therefore, the MBBP decoder does not
have to store all redundant bases, but rather a set of
permutations, along with a low number of cog vectors that are known to have good
stopping set properties. Note that this is a desirable property
for practical applications: in order to generate a set of matrices
with almost identical and good decoding performance, one only
needs to determine the stopping set properties of the cog families
rather than that of all individual cogs.

Let $\hat{f}$ denote the index of the optimal or near-optimal
family. One vector $\cog_{\hat{l}}$,
$\hat{l}\in{\mathcal{F}}_{\hat{f}}$, and a subset of permutations in
$\Aut({\mathcal{C}})$ suffice to generate a set of cogs from
${\mathcal{F}}_{\hat{f}}$. Depending on the set of available
permutations, multiple cogs may be required to generate all cogs from
the family ${\mathcal{F}}_{\hat{f}}$.  With all these cogs at hand, one can construct the
matrices $\ve{H}_{\ell}$, $\ell\in{\mathcal{F}}_{\hat{f}}$, required
for MBBP decoding, by cyclically shifting the corresponding cogs,
cf.\ Section~\ref{sec:definitions}.

For some classes of cyclic codes, stronger results are available
on the structure of the automorphism group. For example, it is
known that the automorphism group of a primitive BCH code contains
the affine group - and BCH cog families are invariant under the
action of every element of the affine group. On the other hand,
the family of binary extended quadratic residue (QR) codes
\cite{macwilliamsetal77} with lengths $p+1$, where $p$ is a prime,
is known to have automorphism groups that contain permutations
from the projective special linear group $PSL(p)$. For $p=8m\pm
1$, the elements of this group can be written as
\begin{equation*}
y \to \frac{a\,y+b}{c\,y+d}, \; a,b,c,d \in \mathbb{F}_p,\;
ad-bc=1.
\end{equation*}
A non-affine permutation in this group transforms cogs from one
family into cogs that belong to a \emph{different} family. As a
result, the image of the original cog generates a parity-check
matrix with different stopping set properties compared to the one
generated by the original cog.

\section{MBBP decoding}

\label{sec:mbbp_and_variations}

We describe next in more detail the operating principles of the
MBBP decoders, and modifications thereof - as depicted in
Figure~\ref{fig:parallel_decoder_structure} and
\ref{fig:parallel_decoder_structure_information_exchange},
respectively. The basic components of an MBBP decoder are
collections of (possibly redundant) parity-check matrices, and
logical units that combine and process outputs of BP decoders
operating on the matrices of the collection. We distinguish two
basic MBBP architectures: one, which allows information exchange
between decoders ({\emph{MBBP-X}}) and another, where the decoder
outputs are obtained without exchange of information
({\emph{MBBP-NX}}). The decoders in the former category have the
feature that information on the reliability of the received
symbols can be exchanged \emph{during} the process of iterative BP
decoding; decoders in the latter class can only combine their
results \emph{upon termination} of their individual decoding
processes. Both types of decoders can be implemented by storing a
set of parity-check matrices. However, if only matrices based on
one family of cogs are used, one needs to store only a low number of cogs,
along with a set of permutations from ${\mathcal{P}}$.

The simplest architecture of an MBBP decoder is depicted in
Figure~\ref{fig:parallel_decoder_structure}, where the outputs of
individual decoders are jointly processed only at the end of the
decoding cycle. We refer to this technique as \emph{standard} MBBP
decoding (henceforth, MBBP-NX-S decoding).

\subsection{MBBP-NX decoding}

Standard MBBP decoding and its variation {\emph{First-success
MBBP decoding}} (MBBP-NX-FS) generate a collection of decoded words and then
perform an additional metric selection within this set of words.
The result of this processing is passed on to the
\emph{information sink}, which represents the gateway for the
final codeword estimate of the decoder.

\subsubsection{Standard MBBP decoding (MBBP-NX-S)}
\label{MBBP-NX-S}

The MBBP-NX-S decoder runs multiple BP decoders in parallel, each
of them separately and on a different parity-check matrix
representation of the code. Let the parity-check matrix
representation used by the $\ell$-th decoder be denoted by
$\ve{H}_{\ell}$, $\ell=1,\dots, l$, and its decoded vector after at most $N$ iterations by $\hat{\ve{c}}_{\ell}$,
$\ell=1,\dots, l$. We let ${\mathcal{V}}\subseteq\{1,\dots, l\}$
be the set of indices $\ell$ describing decoders that converged to
a valid codeword. If none of the decoders converged to a valid
codeword, we let ${\mathcal{V}}=\{1,\dots,l\}$.  The words
estimated by the decoders, $\hat{\ve{c}}_{v}$, $v\in{\mathcal{V}}$
are passed on to a {\emph{least metric selector}} (LMS) unit,
which determines the ``best'' codeword estimate using the decision
rule
$\hat{\ve{c}}=\argmax{v\in{\mathcal{V}}}\Pr\left\{\ve{Y}=\ve{y}
\mid\ve{C}=\hat{\ve{c}}_v\right\}$. It is well known that for the
AWGN channel, this equation can be rewritten~\cite{linetal04} as $
\hat{\ve{c}}=\argmin{v\in{\mathcal{V}}}\sum\limits_{\nu=0}^{n-1}
\left|y_{\nu}-\mapping(\hat{c}_{v,\nu})\right|^2$. In this
context, $\mapping(\cdot)$ defines the mapping of binary input
symbols into antipodal signals (i.e.\ the BPSK modulation performed by
the transmitter). The estimated information vector $\hat{\ve{u}}$
is obtained from $\hat{\ve{c}}$ in the standard manner. We choose
a generator matrix of systematic form what provides several
advantages, described in more detail in a companion paper
\cite{hehnetal08a}. Figure~\ref{fig:parallel_decoder_structure}
depicts the operation of the MBBP-NX-S decoder.

\begin{figure}[ht!]
\begin{center}
\psfrag{rec}[c][c][0.95]{\scriptsize{$\ve{y}$}}
\psfrag{success1}[l][l][0.95]{\scriptsize{$1\!\in\!{\mathcal{V}}$\,\,?}}
\psfrag{success2}[l][l][0.95]{\scriptsize{$2\!\in\!{\mathcal{V}}$\,\,?}}
\psfrag{successl}[l][l][0.95]{\scriptsize{$l\!\in\!{\mathcal{V}}$\,\,?}}
\psfrag{est}[c][c][0.95]{\scriptsize{$\hat{\ve{c}}$}}
\psfrag{BPD}[c][c][0.95]{{\scriptsize{BP-Dec.}}}
\psfrag{on}[c][c][0.95]{\scriptsize{{on}}}
\psfrag{c1}[l][l]{\scriptsize{$\hat{\ve{c}}_1$}}
\psfrag{c2}[l][l]{\scriptsize{$\hat{\ve{c}}_2$}}
\psfrag{cl}[l][l]{\scriptsize{$\hat{\ve{c}}_l$}}
\psfrag{Rep1}[c][c][0.95]{{\scriptsize{$\ve{H}_1$}}}
\psfrag{Rep2}[c][c][0.95]{\scriptsize{{$\ve{H}_2$}}}
\psfrag{Repl}[c][c][0.95]{{\scriptsize{$\ve{H}_l$}}}
\psfrag{dots}[c][c][0.95]{{\scriptsize{$\dots$}}}
\psfrag{ML}[c][c][0.95]{{\scriptsize{$\LMS$}}}
\psfrag{I}[r][r][0.95]{{\scriptsize{$N$}}}
\includegraphics[scale=0.45]{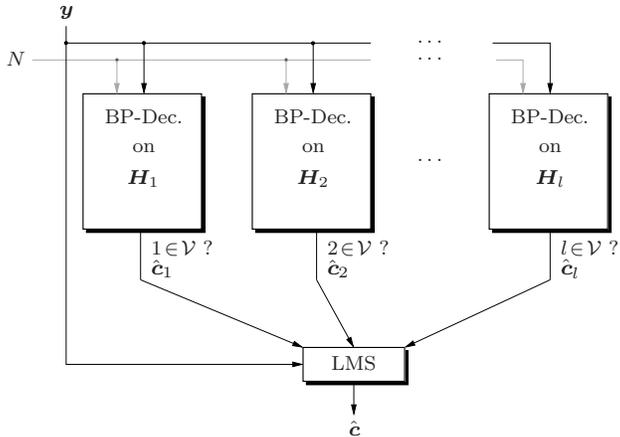}
\caption{MBBP-NX-S decoding}
\label{fig:parallel_decoder_structure}
\end{center}
\vspace{-0.6cm}
\end{figure}

\subsubsection{First-success MBBP decoding (MBBP-NX-FS)}

This type of MBBP decoder follows the standard approach in so far that
it runs multiple BP decoders separately, each on a different
parity-check matrix of the code. Denote the number of iterations
required by the $\ell$-th decoder to converge by $N_{\ell}$. As soon
as the first decoder, indexed by $\ell^{*}$, identifies a codeword,
the decoding process terminates. The estimate obtained by the decoder
indexed by $\ell^{*}$ is passed on to the information sink. In the
unlikely event that two or more decoders converge to a codeword after
the same number of iterations, one of the outputs is randomly chosen
from $\ve{c}_v$, where $v\in{\mathcal{V}}$ and where ${\mathcal{V}}$
is defined as for the MBBP-NX-S decoder. This approach requires only
$N_{\ell^{*}}=\min_{\ell}N_{\ell}$ iterations to converge, and has
therefore a significantly lower time-complexity when compared to
MBBP-NX-S.
 When considering the average number of iterations to decode one codeword, this effect shows in particular in the low signal-to-noise ratio (SNR) regime.

The steps of these algorithms are summarized in Algorithm
\ref{alg:MBBPNX}. Note that the tags of a selected set of steps
and commands in the table indicate the algorithm during which
these steps are executed. All untagged steps are executed both
during the MBBP-NX-S and the MBBP-NX-FS procedure.



Line $2$, $5$, and $6$ in Algorithm \ref{alg:MBBPNX} describe the
actions of detecting the first convergent decoder and terminating
the decoding process. Although it can be easily shown that the
MBBP-NX-FS algorithm cannot outperform its MBBP-NX-S counterpart,
it provides for significant time savings, which makes it amenable
for use in low-delay communication systems.

\begin{algorithm}[ht!]
\begin{algorithmic}[1]
\INPUT{$\ve{y}$, $\ve{H}_1,\dots,\ve{H}_{l}$, $N$}
\OUTPUT{$\hat{\ve{c}}$}
\STATE{${\mathcal{V}}:=\emptyset$, $i:=1$}
\WHILE{$i\leq N$ (NX-S, NX-FS) \textbf{and} ${\mathcal{V}}=\emptyset$ (NX-FS)}
\STATE{$\hat{\ve{c}}_{\ell}:={\mathrm{BPiteration}}(\ve{y},\ve{H}_{\ell})$,
$\ell\in\{1,\dots,l\}\backslash{\mathcal{V}}$}
\FOR{$\ell:=1,\dots,l$}
\IF{$\hat{\ve{c}}_{\ell}\cdot\ve{H}_{\ell}^{\transposed}=0$}
\STATE{${\mathcal{V}}:={\mathcal{V}}\cup\ell$}
\ENDIF
\ENDFOR
\STATE{$i:=i+1$}
\ENDWHILE
\IF{${\mathcal{V}}=\emptyset$}
\STATE{${\mathcal{V}}:=\{1,\dots,l\}$} \ENDIF
\STATE{$\hat{\ve{c}}:=\argmin{\ve{c}_v,\, v\in{\mathcal{V}}}
\sum\limits_{\nu=0}^{n-1}\left|y_{\nu}-\mapping(\hat{c}_{v,\nu}\right)|^2$\hspace*{\fill}(NX-S)}
\STATE{$\hat{\ve{c}}:=\randfun{\ve{c}_v,\,
v\in{\mathcal{V}}}$\hspace*{\fill}(NX-FS)}
\end{algorithmic}
\caption{\label{alg:MBBPNX}MBBP-NX: Standard (NX-S), and First-Success (NX-FS)}
\end{algorithm}

\subsection{MBBP-X decoding}

In this section we present MBBP approaches which allow for
\emph{periodic exchange} of information between decoders
\emph{during} iterative BP decoding. To this end, each decoder
performs independently a given number of iterations, $N_{p}$, and
afterwards exchanges information with other decoders only at
iterations indexed by $\iota\cdot N_{p}$, where $\iota\in\setN$.
The {\emph{intrinsic information}} of a given variable node
depends only on the channel output and is therefore equal for all
decoders. For this reason, the decoders exchange only
{\emph{extrinsic information}} about the variable nodes.
Figure~\ref{fig:parallel_decoder_structure_information_exchange}
depicts the basic architecture of an MBBP-X decoder. As before, $\hat{\ve{u}}$ can be obtained from $\hat{\ve{c}}$ in a straightforward manner.

To emphasize that the messages exchanged between decoders
represent extrinsic information, they are denoted by
$\Prob^{(\ext)}(\cdot)$. As part of their \emph{cooperation scheme}, the
decoders agree on the (extrinsic) probability values for each
variable node. Afterwards, each decoder replaces its own
information about a given variable node with the jointly derived
estimate of the decoders, then calculates the {\emph{a-posteriori
information}}, and resumes decoding. For the purpose of computing
the cooperative extrinsic probability, only a subset of active
decoders ${\mathcal{A}}_{\nu}$, $\nu=0,\dots, n-1$, is used. This
subset is selected in terms of a {\emph{soft-metric}} majority
vote, which is described in more detail below.

\begin{figure}[ht!]
\begin{center}
\psfrag{success1}[l][l][0.95]{\scriptsize{$1\!\in\!{\mathcal{V}}$\,\,?}}
\psfrag{success2}[l][l][0.95]{\scriptsize{$2\!\in\!{\mathcal{V}}$\,\,?}}
\psfrag{successl}[l][l][0.95]{\scriptsize{$l\!\in\!{\mathcal{V}}$\,\,?}}
\psfrag{rec}[c][c]{$\ve{y}$}
\psfrag{c1}[l][l]{\scriptsize{$\hat{\ve{c}}_1$}}
\psfrag{c2}[l][l]{\scriptsize{$\hat{\ve{c}}_2$}}
\psfrag{cl}[l][l]{\scriptsize{$\hat{\ve{c}}_l$}}
\psfrag{est}[c][c][0.95]{$\hat{\ve{c}}$}
\psfrag{BPD}[c][c][0.95]{{\scriptsize{BP-Dec.}}}
\psfrag{on}[c][c][0.95]{\scriptsize{{on}}}
\psfrag{Rep1}[c][c][0.95]{{\scriptsize{$\ve{H}_1$}}}
\psfrag{Rep2}[c][c][0.95]{\scriptsize{{$\ve{H}_2$}}}
\psfrag{Repl}[c][c][0.95]{{\scriptsize{$\ve{H}_l$}}}
\psfrag{dots}[c][c][0.95]{{\scriptsize{$\dots$}}}
\psfrag{I}[r][r][0.95]{{\scriptsize{$N$}}}
\psfrag{Bus}[c][c][0.95]{{\scriptsize{Bus}}}
\psfrag{activate}[l][l][0.95]{{\scriptsize{Activate gateway if ${\mathcal{V}}\not=\emptyset$}}}
\psfrag{rand}[c][c][0.95]{{\scriptsize{$\randfun{\ve{c}_v,\,v\in{\mathcal{V}}}$}}}
\includegraphics[scale=0.45]{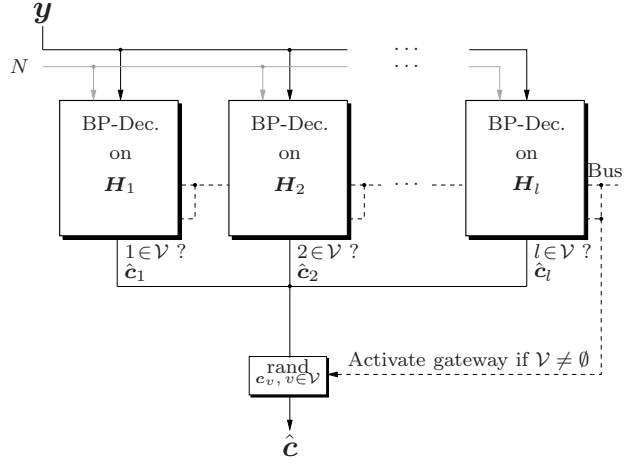}
\caption{MBBP-X decoder architecture. The bus is active only if
$i=\iota\cdot N_p$, $\iota\in\setN$.}
\label{fig:parallel_decoder_structure_information_exchange}
\end{center}
\end{figure}

\subsubsection{Probability-Averaging MBBP (MBBP-X-PA)}

This MBBP variation consists of the same basic steps as the
standard MBBP algorithm, except for a periodic intra-decoder
probability averaging feature. Probabilistic averaging was first
described by the authors in~\cite{laendneretal05}, in the context
of standard BP decoding. There, it was shown that it can lead to
significant reductions in the error-floor of the performance
curve. Here, averaging is used in a different context:
periodically, the parallel decoders update their extrinsic
probabilities by setting them to the average probability of a
subset of the best performing decoders.

For this purpose, a soft-metric, weighted majority-vote
$\bar{v}_{\nu}$ is calculated for a variable node $\nu$ according
to

\begin{equation}
\bar{v}_{\nu}=\sum\limits_{\ell=1}^{l}\log\left(\frac{\Prob^{(\ext)}(c_{\ell,\nu}=0\mid\ve{H}_{\ell},\ve{y})}
{\Prob^{(\ext)}(c_{\ell,\nu}=1\mid\ve{H}_{\ell},\ve{y})}\right),\, \nu=0,\dots,n-1.
\label{eq:soft_metric_majority_vote}
\end{equation}

The subset of decoders used in the described averaging process is
defined as


\begin{equation}
\mathcal{A}_{\nu}=\left\{\ell|\sgn(\bar{v}_{\nu})=
\sgn\left(\log\left(\frac{\Prob^{(\ext)}(c_{\ell,\nu}=0\mid\ve{H}_{\ell},\ve{y})}
{\Prob^{(\ext)}(c_{\ell,\nu}=1\mid\ve{H}_{\ell},\ve{y})}\right)\right)\right\},\,
\nu=0,\dots,n-1, \label{eq:set_of_participating_decoders}
\end{equation}
where $\sgn(\cdot)$ denotes the sign function.

The averaged probabilities are determined as
\begin{equation}
\Prob^{(\ext)}(c_{\nu}=0)=\frac{1}{\left|\mathcal{A}_{\nu}\right|}\sum\limits_{\ell'\in\mathcal{A}_{\nu}}
\Prob^{(\ext)}(c_{\ell',\nu}=0\mid\ve{H}_{\ell'},\ve{y}),\, \nu=0,\dots,n-1,
\label{eq:averaging_1}
\end{equation}
and
\begin{equation}
\Prob^{(\ext)}(c_{\nu}=1)=1-\Prob^{(\ext)}(c_{\nu}=0),\, \nu=0,\dots,n-1.
\label{eq:averaging_4}
\end{equation}

We point out that the performance of this decoding approach
strongly depends on the specific implementation of the proposed steps, as already pointed out in \cite{laendneretal05}.





\subsubsection{Highest-Reliability MBBP (MBBP-X-HR)}

This approach is a simple modification of the MBBP-X-PA technique,
and it employs the weighted majority-vote introduced in Equation
(\ref{eq:soft_metric_majority_vote}) to select the decoders used for
computing the information update subsequently forwarded to all
decoders.

In this approach, the MBBP decoding architecture selects for each variable node
an information update which is the most reliable output
among all decoders, i.e.
\begin{equation}
\Prob^{(\ext)}(c_{\nu}=0)=\Prob^{(\ext)}(c_{\ell^*,\nu}=0\mid \ve{H}_{\ell^*},\ve{y})
\label{eq:highest_reliability}
\end{equation}
with $\ell^*=\argmax{\ell'\in{\mathcal{A}}_{\nu}}\,\,
|\Prob^{(\ext)}(c_{\ell',\nu}=0\mid \ve{H}_{\ell'},\ve{y})-0.5|,\,
\nu=0,\dots,n-1$, and with $A_{\nu}$ given according to Equation
(\ref{eq:set_of_participating_decoders}). Note that
$\Prob^{(\ext)}(c_{\nu}=1)$ is calculated according to
Equation~(\ref{eq:averaging_4}).

\subsubsection{Information-Combining MBBP (MBBP-X-IC)}

The optimal method for deciding on the value of a random variable
when multiple \emph{independent} noisy observations of a
variable are given is {\emph{information combining}}
\cite{landetal06}.

The extrinsic information provided by the parallel decoders
depends on their underlying parity-check matrices and the received
codeword $\ve{y}$, and is hence not independent. Nevertheless,
without assuming any optimality properties, we propose an MBBP
decoding architecture which exchanges joint information about the
variable nodes and determines this value by information combining.
Information combining is also performed separately for each
variable node as part of standard BP, and again only the decoders
which agree with the soft-metric majority-vote according to
Equation (\ref{eq:soft_metric_majority_vote}) are considered. In
other words, only representations that have an index that belongs
to the set ${\mathcal{A}}_{\nu}$, defined as in Equation
(\ref{eq:set_of_participating_decoders}), are used in this
approach.

The ``combined'' probability distribution of the binary variable
$X_{\nu}$ is given by


\begin{equation}
\Prob^{(\ext)}(c_{\nu}=0)=\frac{\prod\limits_{\ell'\in{\mathcal{A}}_{\nu}}\Prob^{(\ext)}(c_{\ell',\nu}=0\mid \ve{H}_{\ell'},\ve{y})}
{\prod\limits_{\ell'\in{\mathcal{A}}_{\nu}}\Prob^{(\ext)}(c_{\ell',\nu}=0\mid \ve{H}_{\ell'},\ve{y})
+\prod\limits_{\ell'\in{\mathcal{A}}_{\nu}}\Prob^{(\ext)}(c_{\ell',\nu}=1\mid \ve{H}_{\ell'},\ve{y})},\, \nu=0,\dots,n-1.
\label{eq:info_comb_formula}
\end{equation}
Again, one can use Equation~(\ref{eq:averaging_4}) to calculate $\Prob^{(\ext)}(c_{\nu}=1)$.

A summary of the steps for the three algorithms from the class of
MBBP-X decoders is given in Algorithm \ref{alg:MBBPX}. As before, the
untagged lines are executed in all three algorithms, while the
tags of the lines specify if the step is to be used in the X-PA,
X-IC, or X-HR variation of MBBP decoding.

\begin{algorithm}[ht!]
\begin{algorithmic}[1]
\INPUT{$\ve{y}$, $\ve{H}_1,\dots,\ve{H}_{l}$, $N$, $N_p$}
\OUTPUT{$\hat{\ve{c}}$} \STATE{${\mathcal{V}}:=\emptyset$, $i:=1$}
\STATE{$\Prob^{(\ext)}(c_{\nu}\hspace{-0.1cm}=\hspace{-0.1cm}0\hspace{-0.1cm}\mid\hspace{-0.1cm}\ve{H}_{\ell},\ve{y})=0.5$,
$\nu=0,\dots,n-1$, $\ell=1,\dots,l$} \WHILE{$i\leq N$ and
${\mathcal{V}}=\emptyset$}
\STATE{$[\hat{\ve{c}}_{\ell},\Prob^{(\ext)}(c_{\ell,\nu}=0\mid\ve{H}_{\ell},\ve{y})]:={\mathrm{BPiteration}}(\ve{y},\ve{H}_{\ell},\Prob^{(\ext)}(c_{\nu}\hspace{-0.1cm}=\hspace{-0.1cm}0\hspace{-0.1cm}\mid\hspace{-0.1cm}
\ve{H}_{\ell},\ve{y}))$, $\nu=0,\dots,n-1$, $\ell=1,\dots,l$}
\IF{$i=\iota\cdot N_p$, $\iota\in\setN$} \STATE{Apply Eq.\
(\ref{eq:soft_metric_majority_vote}) $\rightarrow\bar{v}_{\nu}$}
\STATE{Apply Eq.\ (\ref{eq:set_of_participating_decoders})
$\rightarrow{\mathcal{A}}_{\nu}$} \STATE{Apply Eqs.\
(\ref{eq:averaging_1}) and (\ref{eq:averaging_4})
$\rightarrow\Prob^{(\ext)}(c_{\nu}=0/1),\,
\nu=0,\dots,n-1$\hspace*{\fill}(X-PA) } \STATE{Apply Eq.\
(\ref{eq:highest_reliability})$\rightarrow\Prob^{(\ext)}(c_{\nu}=0/1),\,
\nu=0,\dots,n-1$\hspace*{\fill}(X-HR)} \STATE{Apply Eq.\
(\ref{eq:info_comb_formula})$\rightarrow\Prob^{(\ext)}(c_{\nu}=0/1),\,
\nu=0,\dots,n-1$\hspace*{\fill}(X-IC)}
\STATE{$\Prob^{(\ext)}(c_{\nu}\hspace{-0.1cm}=\hspace{-0.1cm}0\hspace{-0.1cm}\mid\hspace{-0.1cm}
\ve{H}_{\ell},\ve{y})\hspace{-0.1cm}:=\hspace{-0.1cm}\Prob^{(\ext)}(c_{\nu}=0)$,
$\nu=0,\dots,n-1$, $\ell=1,\dots,l$} \FOR{$\ell:=1,\dots,l$}
\IF{$\hat{\ve{c}}_{\ell}\cdot\ve{H}_{\ell}^{\transposed}=0$}
\STATE{${\mathcal{V}}:={\mathcal{V}}\cup\ell$}
\ENDIF
\ENDFOR
\ENDIF
\STATE{$i:=i+1$}
\ENDWHILE
\IF{${\mathcal{V}}=\emptyset$}
\STATE{${\mathcal{V}}:=\{1,\dots,l\}$}
\ENDIF
\STATE{$\hat{\ve{c}}:=\randfun{\ve{c}_v,\, v\in{\mathcal{V}}}$}
\end{algorithmic}
\caption{\label{alg:MBBPX}MBBP-X: Probability-Averaging (PA), Information-Combining (IC), and Highest-Reliability (HR)}
\end{algorithm}

\section{Results}
\label{sec:results}

As already pointed out, parity-check matrices that offer good
performance when used for decoding signals transmitted over the
BEC may also be good candidates for BP decoding over the AWGN
channel \cite{richardsonetal01a}. The same characteristic of
parity-check matrices was observed through extensive computer
simulations for the variants of MBBP decoders, used over the AWGN
channel.

In this section, we present performance results for MBBP decoders.
Whenever possible, we provide a comparison of the error rates of
MBBP decoders with those of a full search algorithm
(maximum-likelihood, ML, decoding). If performing full search is
computationally infeasible, but results on the weight distribution
$\{A_i\},i=0,1,\ldots,n$, of the underlying code are available, we
plot the union bounds for the $\BER$ and $\FER$ as given by the
expressions below: \
\[
\BER\leq \frac{1}{n}\sum_{i=d}^{n}A_{i}\cdot i \cdot
Q\left(\sqrt{2\frac{k}{n}i\frac{\Eb}{N_0}}\right),
\]

\[
\FER\leq \sum_{i=d}^{n}A_{i}\cdot
Q\left(\sqrt{2\frac{k}{n}i\frac{\Eb}{N_0}}\right).
\]
Here, it is tacitly assumed that the transmitted codeword is the
all-zero word, and $A_{i}$ is used to denote the number of
codewords of weight $i$ in the given code. Furthermore, we compare
our results to the {\emph{Gallager bound}} (random coding bound)
\cite{gallager68}. This bound considers an {\emph{ensemble}} of
block codes with length $n$ and rate $R$ and provides a tight
upper bound on the average $\FER$, denoted by $\Expect\{\FER\}$.
In this context, the expectation is taken over the ensemble of
codes. According to~\cite{gallager68}, the upper bound reads

\begin{equation}
\Expect\{\FER\}\leq \exp\left(-n \cdot \max_{0\leq\rho\leq 1}\max_{\Pr(\ve{X})} \left(E_0(\rho,10\log_{10}(E_{\mathrm{b}}/N_0),\Prob(\ve{X}))-\rho R\right)\right),
\label{eq:exponent_bound_ready}
\end{equation}
where $\rho$, $0\leq\rho\leq 1$, is a design parameter and $\Prob(\ve{X})$ denotes the probability vector of the possible channel inputs.
Let us specialize this bound for the AWGN channel, which has
discrete inputs and an SNR-dependent transition density function
$f_y(y\mid x_i)$. Here, $x_i$ is chosen from a finite set of
cardinality $M_x$ and $y$ is a continuous variable, and thus

\[
E_0(\rho,10\log_{10}(E_{\mathrm{b}}/N_0), \Prob(\ve{X}))=-\ln\left[\int\limits_{-\infty}^{\infty}
\left(\sum\limits_{i=1}^{M_x}[f_y(y\mid x_i)]^{\frac{1}{1+\rho}}\Prob(x_i)\right)^{1+\rho}\,{\mathrm{d}}y\right].
\]

We illustrate our findings on four short-to-moderate length codes;
these codes belong to the class of cyclic and extended cyclic codes,
and they include the $[24,12,8]$ extended Golay code, the $[47,24,11]$ quadratic residue (QR) code, the $[31,16,7]$-BCH code, as well as the
$[127,64,21]$-BCH code.
The maximum number of iterations is set to $N=100$ for all codes
and decoding approaches considered. Furthermore, $N_p$ is set to
$10$ whenever MBBP-X approaches are simulated.

Throughout the remainder of this
section, $l$ is used to denote the number of parallel BP decoders
in the MBBP architecture. For comparison, simulation results for standard BP decoding are presented as well. Also shown are the performance results of a BP decoder using the union of all parity-check equations in $\ve{H}_{\ell},\,\ell=1,\dots,l$, simultaneously, for the case of the $[31,16,7]$-BCH code.

\subsection{The $[24,12,8]$ Extended Golay Code}
\label{sec:golay}

For the purpose of MBBP decoding of the extended Golay code, we
use the result of Sections \ref{sec:definitions} and
\ref{sec:bases_selection}, and identify three different cog
families denoted by ${\mathcal{F}}_1$, ${\mathcal{F}}_2$, and
${\mathcal{F}}_3$. For this, and all subsequently considered
codes, one can generate cogs of a family by repeated application
of permutations of type $P_2$, shown in
Equation~(\ref{eq:perm_i_to_2i}). As an example, we will describe
the process of constructing the parity-check matrices for the
extended Golay code in more detail.

Since the code is an extended cyclic code, we construct the
parity-check matrices from each cog in terms of $23$ shifts
performed on positions $0$ to $22$ while keeping the last position
fixed.  For the $24$-th row of the parity-check matrix, we use the
all-one codeword: this codeword preserves the stopping set
distribution of the $23 \times 24$ matrix, and is the only parity
check of the $[24,12,8]$ extended Golay code invariant under all
affine permutations. It is worth pointing out that there are other
choices for the $24$-th parity-check equation that may lead to
slightly better overall performance - we use this parity-check
matrix for simplicity of analysis. As a performance criterion for
the cog families, we use the number of stopping sets up to size
$d=8$ in the parity-check matrices $\ve{H}_{\ell}$,
$\ell\in\{{\mathcal{F}}_1,{\mathcal{F}}_2,{\mathcal{F}}_3\}$.

We consider one representative cog for each family, termed
$\cog_{\mathcal{F}_f}$, $f=1,2,3$. A list of these cogs is given
below.

\begin{eqnarray*}
\cog_{\mathcal{F}_1}&=&1\,1\,0\,1\,0\,1\,0\,0\,1\,1\,0\,0\,1\,0\,0\,0\,0\,0\,0\,0\,1\,0\,0\,0\\
\cog_{\mathcal{F}_2}&=&1\,1\,1\,0\,0\,0\,0\,0\,1\,0\,0\,1\,1\,0\,0\,0\,0\,0\,1\,0\,0\,0\,0\,1\\
\cog_{\mathcal{F}_3}&=&1\,1\,0\,1\,0\,0\,1\,1\,0\,0\,0\,0\,0\,0\,0\,1\,0\,1\,0\,0\,1\,0\,0\,0
\end{eqnarray*}

Each $\cog_{\mathcal{F}_f}$, $f\in\{1,2,3\}$, allows for
generating all cogs in the family ${\mathcal{F}}_f$ by repeated
application of permutation $P_2$ to the first $23$ positions of
$\cog_{\mathcal{F}_f}$ while keeping the last position fixed, cf.\
Equation (\ref{eq:perm_i_to_2i}). The permutation $P_2$ is an
automorphism of an extended cyclic code and preserves the stopping
set distribution of the matrices $\ve{H}_{\ell}$,
$\ell\in{\mathcal{F}}_f$, cf.\ Claim~\ref{claim:permutation_beta}.
The number of stopping sets in $\ve{H}_{\ell}$,
$\ell\in\{{\mathcal{F}}_1,{\mathcal{F}}_2,{\mathcal{F}}_3\}$, is
summarized in Table
\ref{table:unresolved_stopping_sets_golay_24_12}. Note that the
matrices of the considered families differ significantly from the
near-optimal and highly redundant matrices used for decoding over the BEC, given in
\cite{schwartzetal06} and~\cite{hanetal07}. Due to the
high redundancy, these matrices are not amenable for
decoding over the AWGN channel, where short cycles may
significantly degrade the performance of iterative decoders.

\begin{table}
\begin{center}
\begin{tabular}{c|ccc}\hline
&\multicolumn{3}{c}{$\left|{\mathcal{S}}_{\sigma}(\ve{H}_{\ell})\right|$, $\ell\in$}\\
&${\mathcal{F}}_1$&${\mathcal{F}}_2$&${\mathcal{F}}_3$\\\hline
$\sigma\leq 5$&$0$&$0$&$0$\\
$\sigma=6$&$0$&$437$&$46$\\
$\sigma=7$&$1357$&$10143$&$1495$\\
$\sigma=8$&$25783$&$73209$&$20631$\\\hline
\end{tabular}
\end{center}
\caption{\label{table:unresolved_stopping_sets_golay_24_12}Number of
stopping sets for the $[24,12,8]$ extended Golay code, with
parity-check matrices $\ve{H}_{\ell}$,
$\cog_{\ell}\in{\mathcal{F}}_f$, $f=1,\dots,3$.}
\end{table}

The $[24,12,8]$ extended Golay code is self-dual and contains $759$ codewords of minimum
weight $8$ which can be partitioned into $33$ cyclic orbits.
Again, these cyclic orbits can be partitioned into three families
of equal size. In other words, each family ${\mathcal{F}}_f$,
$f=1,2,3$, contains $11$ cogs.

If matrices from the same family are used for signaling over the BEC,
they provide the same performance under iterative
decoding. Interestingly, the simulation results presented below show that the
same is true for decoders used over the AWGN channel.

Figure~\ref{fig:golay_24_12_oversized_parallel_vs_iterations_mode_11_decoders_all_families_all_one_row_100_iter_max_100_iter}
shows the $\BER$ performance of MBBP-NX-S decoding and MBBP-NX-FS
decoding as well as MBBP-X-PA, MBBP-X-HR, and MBBP-X-IC of the
extended Golay code. For all decoder types, $l=11$ parallel BP
decoders are used, for which the parity-check matrices are all
drawn from the same family. $\FER$ performance results show
similar characteristics. These results are not plotted due to space limitations.

\begin{figure}[ht!]
\begin{center}
\subfigure[$\BER$ performance of MBBP-NX approaches]{
\psfrag{BER}[cb][cb]{$\BER\,\rightarrow$}
\psfrag{SNR}[ct][ct]{$10\log_{10}(\Eb/N_0)\,\rightarrow$}
\psfrag{ber 24 12 fam 1}[l][l][1]{\scriptsize{$\ve{H}_{\ell}$, $\ell\in{\mathcal{F}}_1$}}
\psfrag{ber 24 12 fam 2}[l][l][1]{\scriptsize{$\ve{H}_{\ell}$, $\ell\in{\mathcal{F}}_2$}}
\psfrag{ber 24 12 fam 3}[l][l][1]{\scriptsize{$\ve{H}_{\ell}$, $\ell\in{\mathcal{F}}_3$}}
\psfrag{ber 24 12 BP}[l][l][1]{\scriptsize{BP}}
\psfrag{ber 24 12 mbbp-nx s}[l][l][1]{\scriptsize{MBBP-NX-S, $l=11$}}
\psfrag{ber 24 12 mbbp-nx iter}[l][l][1]{\scriptsize{MBBP-NX-FS, $l=11$}}
\psfrag{ber 24 12 MLSE}[l][l][1]{\scriptsize{ML decoding}}
\psfrag{ber 24 12 gallager bound}[l][l][1]{\scriptsize{Gallager bound}}
\includegraphics[scale=0.48]{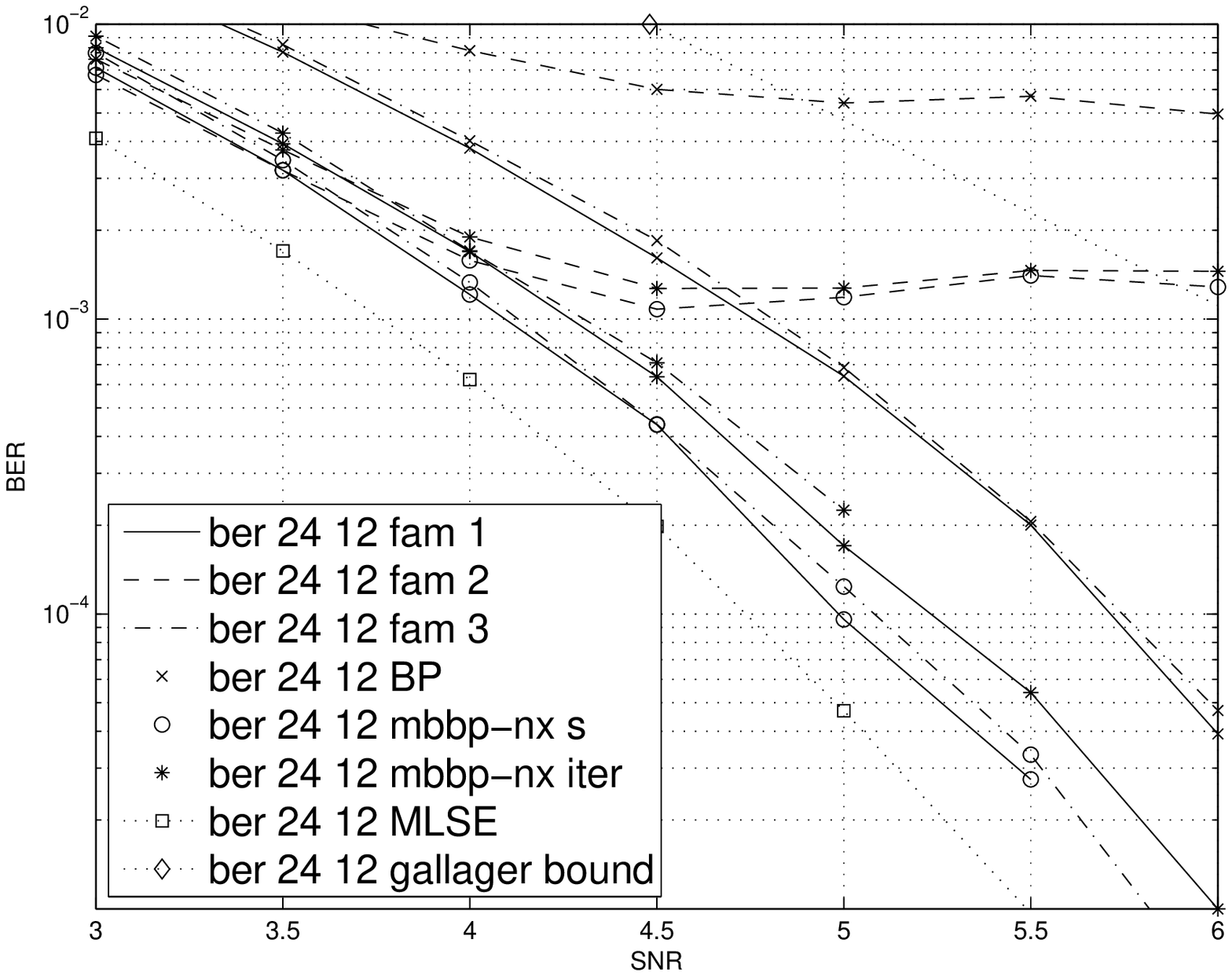}}
\subfigure[$\BER$ performance of MBBP-X approaches]{
\psfrag{BER}[cb][cb]{$\BER\,\rightarrow$}
\psfrag{SNR}[ct][ct]{$10\log_{10}(\Eb/N_0)\,\rightarrow$}
\psfrag{ber 24 12 fam 1}[l][l][1]{\scriptsize{$\ell\!\in\!{\mathcal{F}}_1$}}
\psfrag{ber 24 12 fam 2}[l][l][1]{\scriptsize{$\ell\!\in\!{\mathcal{F}}_2$}}
\psfrag{ber 24 12 fam 3}[l][l][1]{\scriptsize{$\ell\!\in\!{\mathcal{F}}_3$}}
\psfrag{ber 24 12 bp}[l][l][1]{\scriptsize{BP}}
\psfrag{ber 24 12 mbbp-x pa__}[l][l][1]{\scriptsize{MBBP-X-PA, $l=11$}}
\psfrag{ber 24 12 mbbp-x hr}[l][l][1]{\scriptsize{MBBP-X-HR, $l=11$}}
\psfrag{ber 24 12 mbbp-x ic}[l][l][1]{\scriptsize{MBBP-X-IC, $l=11$}}
\psfrag{ber 24 12 MLSE}[l][l][1]{\scriptsize{ML decoding}}
\psfrag{ber 24 12 gallager bound}[l][l][1]{\scriptsize{Gallager bound}}
\includegraphics[scale=0.48]{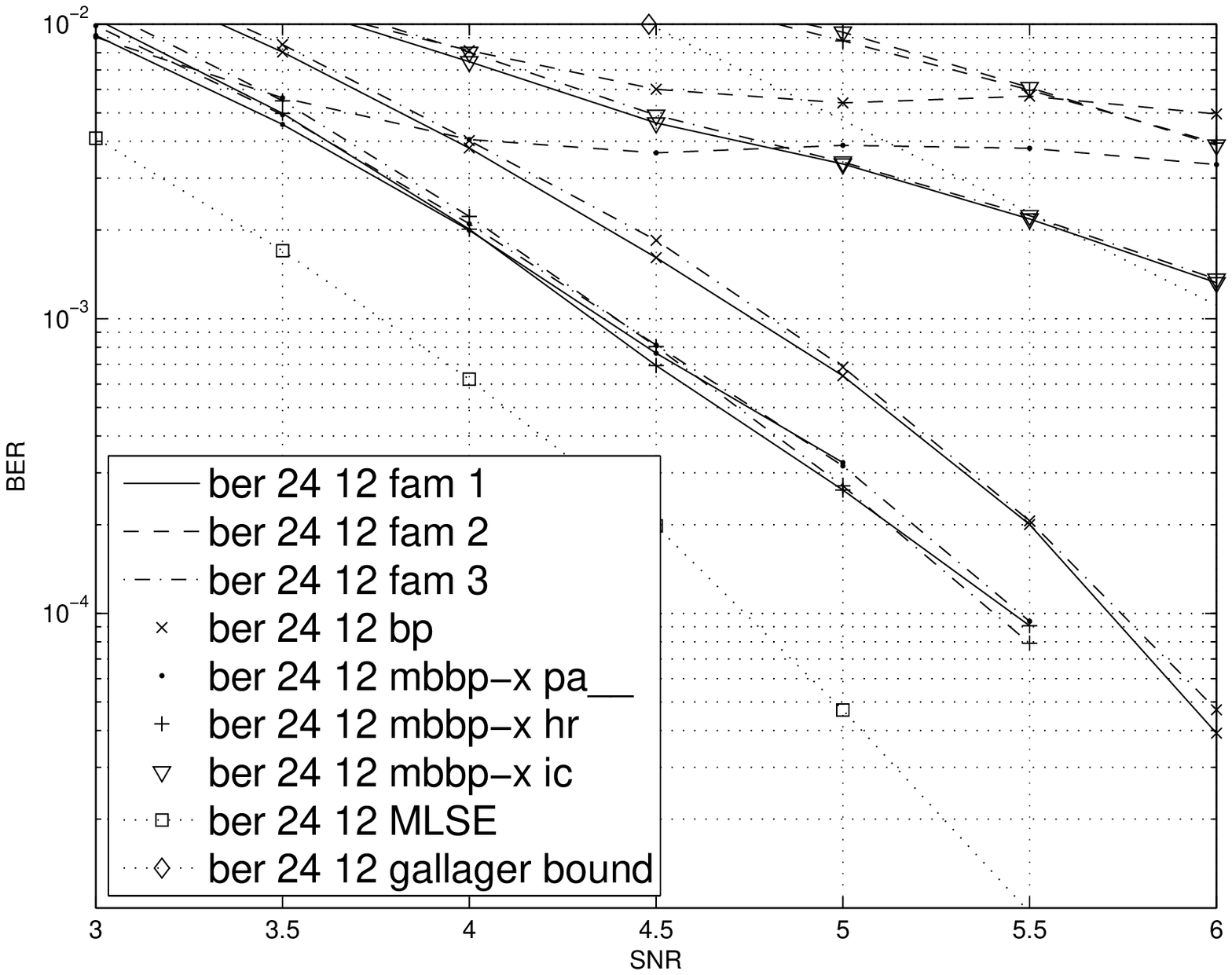}}
\caption{\label{fig:golay_24_12_oversized_parallel_vs_iterations_mode_11_decoders_all_families_all_one_row_100_iter_max_100_iter}Performance comparison for the $[24,12,8]$ extended Golay code using $\ve{H}_{\ell}$, $\ell\in{\mathcal{F}}_1$, ${\mathcal{F}}_2$, ${\mathcal{F}}_3$.}
\end{center}
\end{figure}

It can be observed for all approaches that MBBP decoders using
matrices $\ve{H}_{\ell}$, $\ell\in{\mathcal{F}}_1$, have the best
performance, followed by matrices $\ve{H}_{\ell}$,
$\ell\in{\mathcal{F}}_3$. This finding is supported by the stopping
set distribution of Table
\ref{table:unresolved_stopping_sets_golay_24_12}, showing that the
first family contains matrices that do not have stopping sets of size
up to six and a low number of stopping sets of size seven. Matrices
from the other two families have stopping sets of size six and exceed
the number of stopping sets of size seven of the first family. The
performance of MBBP decoders using parity-check matrices
$\ve{H}_{\ell}$, $\ell\in{\mathcal{F}}_2$, is
significantly worse than that of the two other classes - matrices in
this family have $437$ stopping sets of size six and over $10,000$
stopping sets of size seven.

The performance obtained by means of the MBBP-NX-S with $l=11$ is approximately $0.75$ dB better than standard BP and performs close to the ML decoding bound. 
When performing a direct comparison of MBBP-NX-S and MBBP-NX-FS
decoders, cf.\
Figure~\ref{fig:golay_24_12_oversized_parallel_vs_iterations_mode_11_decoders_all_families_all_one_row_100_iter_max_100_iter},
one can observe that MBBP-NX-FS follows the performance of MBBP-NX-S
very closely. However,
simulation results show a significant gap in the average number of
iterations required for successful decoding when comparing these
approaches. As the shapes of these curves show similar
characteristics, and due to space limitations, these results are
only plotted for the $[127,64,21]$-BCH code. It was observed for
all codes discussed within this paper that MBBP-NX-FS requires
significantly fewer iterations per decoder than MBBP-NX-S and BP decoding.

It is worth pointing out that the MBBP-X-PA and MBBP-X-HR
algorithms produce similar results for all possible choices of cog
families. Also, if the cogs are chosen from ${\mathcal{F}}_1$ or
${\mathcal{F}}_3$, the MBBP-X approaches outperform standard BP
decoding, but do not attain the performance of MBBP-NX decoders.
If the decoders operate on parity-check matrices constructed from
cogs in ${\mathcal{F}}_2$, very poor performance results and error
floors are observed in most of the cases. MBBP-X-IC decoders
perform very poorly, regardless of the family considered. A
probable cause for this phenomena is the strong correlation of the
data, which makes information combining techniques highly
suboptimal.

\subsection{$[31,16,7]$-BCH code}

The dual of the $[31,16,7]$-BCH code, denoted by
${\mathcal{B}}_1^{\perp}$, has minimum Hamming distance equal to
eight. The minimum-weight codewords in ${\mathcal{B}}_1^{\perp}$ can
be partitioned into $15$ cyclic orbits. The corresponding cogs
belong to one single family - as a result, all cog-based parity-check matrices of cyclic form with the same number of rows are expected to
have comparable performance.
In this special case we require $3$ cogs, all from the same family, to generate the $15$ possible cogs. A list of these cogs is given
below.

\begin{eqnarray*}
\cog_{\mathcal{F}_1,1}&=&1\,1\,1\,1\,0\,0\,0\,0\,1\,0\,0\,1\,1\,0\,0\,0\,0\,0\,0\,0\,0\,0\,0\,0\,0\,0\\
\cog_{\mathcal{F}_1,2}&=&1\,1\,0\,0\,0\,0\,1\,1\,0\,0\,0\,0\,0\,1\,0\,1\,0\,0\,0\,0\,0\,0\,0\,0\,0\,0\\
\cog_{\mathcal{F}_1,3}&=&1\,1\,1\,0\,0\,0\,0\,0\,0\,0\,0\,1\,0\,0\,1\,0\,1\,0\,0\,0\,0\,1\,0\,0\,1\,0
\end{eqnarray*}

Figure~\ref{fig:bch_31_16_oversized_parallel_vs_iterations_mode_6_decoders_100_iter_max_100_iter} visualizes the performance results for MBBP decoding of the $[31,16,7]$-BCH code.
We used $l=6$ decoders for the MBBP-NX, MBBP-X-PA, MBBP-X-HR, and
MBBP-X-IC algorithms.

\begin{figure}[ht!]
\begin{center}
\subfigure{
\psfrag{BER}[cb][cb]{$\BER\,\rightarrow$}
\psfrag{SNR}[ct][ct]{$10\log_{10}(\Eb/N_0)\,\rightarrow$}
\psfrag{ber 31 16 representations sim}[l][l]{\scriptsize{Stacked matrix, st.\ BP}}
\psfrag{ber 31 16 representations}[l][l]{\scriptsize{Standard BP decoding}}
\psfrag{ber 31 16 mbbp-nx s}[l][l]{\scriptsize{MBBP-NX-S, $l=6$}}
\psfrag{ber 31 16 mbbp-nx iter}[l][l]{\scriptsize{MBBP-NX-FS, $l=6$}}
\psfrag{ber 31 16 mbbp-x average}[l][l]{\scriptsize{MBBP-X-PA, $l=6$}}
\psfrag{ber 31 16 mbbp-x infocomb}[l][l]{\scriptsize{MBBP-X-IC, $l=6$}}
\psfrag{ber 31 16 mbbp-x maximum}[l][l]{\scriptsize{MBBP-X-HR, $l=6$}}
\psfrag{ber 31 16 MLSE union bound}[l][l]{\scriptsize{ML Union bound}}
\psfrag{ber 31 16 MLSE}[l][l]{\scriptsize{ML}}
\psfrag{ber 31 16 gallager bound}[l][l]{\scriptsize{Gallager bound}}
\includegraphics[scale=0.48]{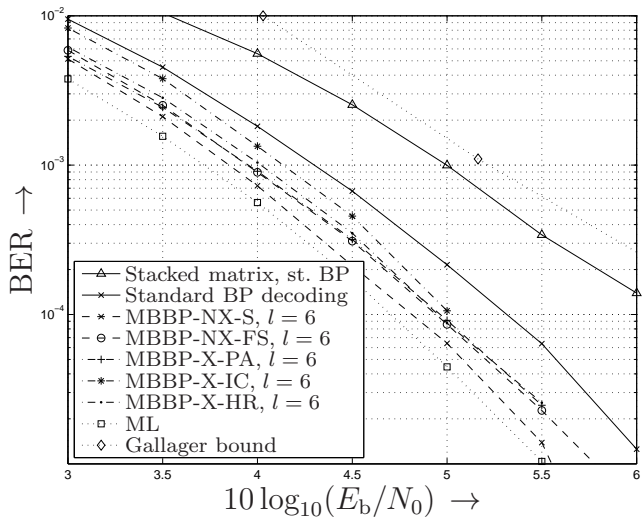}
}
\caption{\label{fig:bch_31_16_oversized_parallel_vs_iterations_mode_6_decoders_100_iter_max_100_iter}
Performance comparison for the
$[31,16,7]$-BCH code.}
\end{center}
\end{figure}

It can be observed that the MBBP-X techniques and the MBBP-NX-FS
decoders have inferior performance compared to the MBBP-NX-S
approach. Yet, these approaches require fewer iterations when more
decoders run in parallel, i.e.\ they usually operate with lower
complexity. It is worth mentioning that the NX-FS approach
exhibits better performance than standard BP decoding and requires
fewer iterations for successful decoding.

The performance curve labeled ``stacked''  corresponds to a BP
decoder operating on a parity-check matrix which contains the
union of parity checks present in $\ve{H}_{\ell}$, $\ell=1,\dots,
6$. Observe that this performance is significantly worse even when
compared to a standard BP decoder running on $\ve{H}_{\ell}$,
$\ell\in\{1,\dots, 6\}$. Reasons for this include the local cycle
distribution of the stacked matrix. All matrices $\ve{H}_{\ell}$,
$\ell=1,\dots, 6$, have a large number of short cycles, and the
stacked matrix has an even larger number of such cycles, and is
therefore not a good candidate for BP decoding.

\subsection{$[47,24,11]$-QR code}

We present next MBBP simulation results for the $[47,24,11]$-QR
code, henceforth denoted by ${\mathcal{Q}}$. The code
${\mathcal{Q}}^{\perp}$ has minimum Hamming distance $12$, and
there exist $276$ codewords of this weight. These codewords can be
partitioned into eight cog families, labeled ${\mathcal{F}}_1$ to
${\mathcal{F}}_8$. In order to identify the family with best
stopping set properties, we compute the number of stopping sets of
size up to and including $\sigma=9$ in cog-based matrices of each
family. Table \ref{tab:unresolved_stopping_sets_qr_47_24} allows
for identifying the family with the best stopping set properties,
${\mathcal{F}}_3$. This family contains $23$ cogs.
Figure~\ref{fig:qr_47_24_oversized_parallel_vs_iterations_mode_23_decoders_group_C_100_iter_max}
plots simulation results for $l=23$ representations with cogs
chosen from ${\mathcal{F}}_3$.

\begin{table}[ht!]
\begin{center}
\begin{tabular}{c|cccccccc}\hline
&\multicolumn{7}{c}{$\left|{\mathcal{S}}_{\sigma}(\ve{H}_{\ell})\right|$, $\ell\in$}\\
&${\mathcal{F}}_1$&${\mathcal{F}}_2$&${\mathcal{F}}_3$&${\mathcal{F}}_4$&${\mathcal{F}}_5$
&${\mathcal{F}}_6$&${\mathcal{F}}_7$&${\mathcal{F}}_8$\\\hline
$\sigma\leq 7$&$0$&$0$&$0$&$0$&$0$&$0$&$0$&$0$\\
$\sigma=8$&$47$&$0$&$0$&$47$&$47$&$47$&$94$&$47$\\
$\sigma=9$&$\times$&$2961$&$2209$&$\times$&$\times$&$\times$&$\times$&$\times$\\\hline
\end{tabular}
\end{center}
\caption{\label{tab:unresolved_stopping_sets_qr_47_24}Number of
stopping sets for parity-check matrices of cyclic form for the
$[47,24,11]$-QR code.}
\end{table}

\begin{figure}[ht!]
\begin{center}
\psfrag{BER}[cb][cb]{$\BER\,\rightarrow$}
\psfrag{SNR}[ct][ct]{$10\log_{10}(\Eb/N_0)\,\rightarrow$}
\psfrag{ber 47 24 representations}[l][l]{\scriptsize{Standard BP decoding}}
\psfrag{ber 47 24 mbbp-nx s}[l][l]{\scriptsize{MBBP-NX-S, $l=23$}}
\psfrag{ber 47 24 mbbp-nx iter}[l][l]{\scriptsize{MBBP-NX-FS, $l=23$}}
\psfrag{ber 47 24 mbbp-x average}[l][l]{\scriptsize{MBBP-X-PA, $l=23$}}
\psfrag{ber 47 24 mbbp-x infocomb}[l][l]{\scriptsize{MBBP-X-IC, $l=23$}}
\psfrag{ber 47 24 mbbp-x maximum}[l][l]{\scriptsize{MBBP-X-HR, $l=23$}}
\psfrag{ber 47 24 MLSE union bound}[l][l]{\scriptsize{ML Union bound}}
\psfrag{ber 47 24 MLSE}[l][l]{\scriptsize{ML}}
\psfrag{ber 47 24 gallager bound}[l][l]{\scriptsize{Gallager bound}}
\includegraphics[scale=0.48]{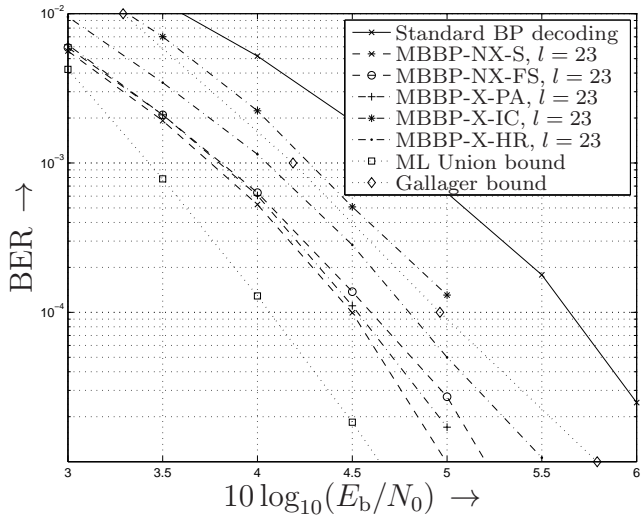}
\caption{\label{fig:qr_47_24_oversized_parallel_vs_iterations_mode_23_decoders_group_C_100_iter_max}
Performance comparison for the $[47,24,11]$-QR code.}
\end{center}
\end{figure}

We list the results for the $[47,24,11]$-QR code as these results allow to
point out two remarkable properties of MBBP decoding. First, we observe that
MBBP decoding can provide a decoding performance which is superior to
the performance reported by the Gallager bound. This is remarkable, as the Gallager bound indicates desirable performance
results for codes of given length and rate. Compared to this bound, standard BP decoding obtains by far poorer performance results. Second, the
$[47,24,11]$-QR code shows that non-standard MBBP decoders can
outperform their standard counterparts.  The approaches MBBP-X-PA and
MBBP-X-HR lead to performance results close to MBBP-NX-S, while the
approach MBBP-X-PA even outperforms MBBP-NX-FS. As in the previous
examples, the MBBP-X-IC decoder suffers from the introduced dependencies that arise due to data exchange between the decoders via
information combining. For this reason, the latter approach shows a
slightly degraded performance, yet still significantly better than
that of standard BP.

\subsection{$[127,64,21]$-BCH code}

The $[127,64,21]$-BCH code, denoted as ${\mathcal{B}}_2$, has
dimension $64$ and a co-dimension $63$. It is prohibitively costly
to decode this code by means of an ML decoder or a trellis
decoder. Also, it is computationally infeasible to find all
minimum weight codewords of the dual code. Instead, we apply the
algorithms provided in~\cite{huetal04} to identify a tight upper
bound on the minimum distance of the dual code as well as a (most
likely incomplete) set of codewords with weight equal to the
estimated minimum distance of the dual code. We report that $22$
is an upper bound on the minimum distance of
${\mathcal{B}}_2^{\perp}$, and that
at least $21$ cog families exist. As it is a very complex task to
identify the number of stopping sets up to the size $d=21$ in a code
of length $127$, we approximate the performance criterion by
evaluating only the number of stopping sets up to $\sigma=5$. This
allows us to conclude that all families yield comparable decoding
performance.

Figure~\ref{fig:bch_127_64_oversized_parallel_vs_iterations_mode_10_decoders_100_iter_max_100_iter}
shows the $\BER$ performance for different variants of the MBBP
approach, where $l=10$ decoders were run in parallel. There exist
similar observations for the $\FER$, but these results are not
shown due to space limitations. Also given is the number of
iterations required for the convergence of different decoding
approaches. In this context, we distinguish between the average
number of iterations required and the maximum number of iterations
required to decode one codeword. It is also worth pointing out
that these results are representative for all codes discussed in this work.

\begin{figure}[ht!]
\begin{center}
\subfigure[$\BER$ performance]{
\psfrag{BER}[cb][cb]{$\BER\,\rightarrow$}
\psfrag{SNR}[ct][ct]{$10\log_{10}(\Eb/N_0)\,\rightarrow$}
\psfrag{ber 127 64 representations___}[l][l]{\scriptsize{Standard BP decoding}}
\psfrag{ber 127 64 mbbp-nx s}[l][l]{\scriptsize{MBBP-NX-S, $l=10$}}
\psfrag{ber 127 64 mbbp-nx iter}[l][l]{\scriptsize{MBBP-NX-FS, $l=10$}}
\psfrag{ber 127 64 mbbp-x average}[l][l]{\scriptsize{MBBP-X-PA, $l=10$}}
\psfrag{ber 127 64 mbbp-x infocomb}[l][l]{\scriptsize{MBBP-X-IC, $l=10$}}
\psfrag{ber 127 64 mbbp-x maximum}[l][l]{\scriptsize{MBBP-X-HR, $l=10$}}
\psfrag{ber 127 64 gallager bound}[l][l]{\scriptsize{Gallager bound}}
\includegraphics[scale=0.48]{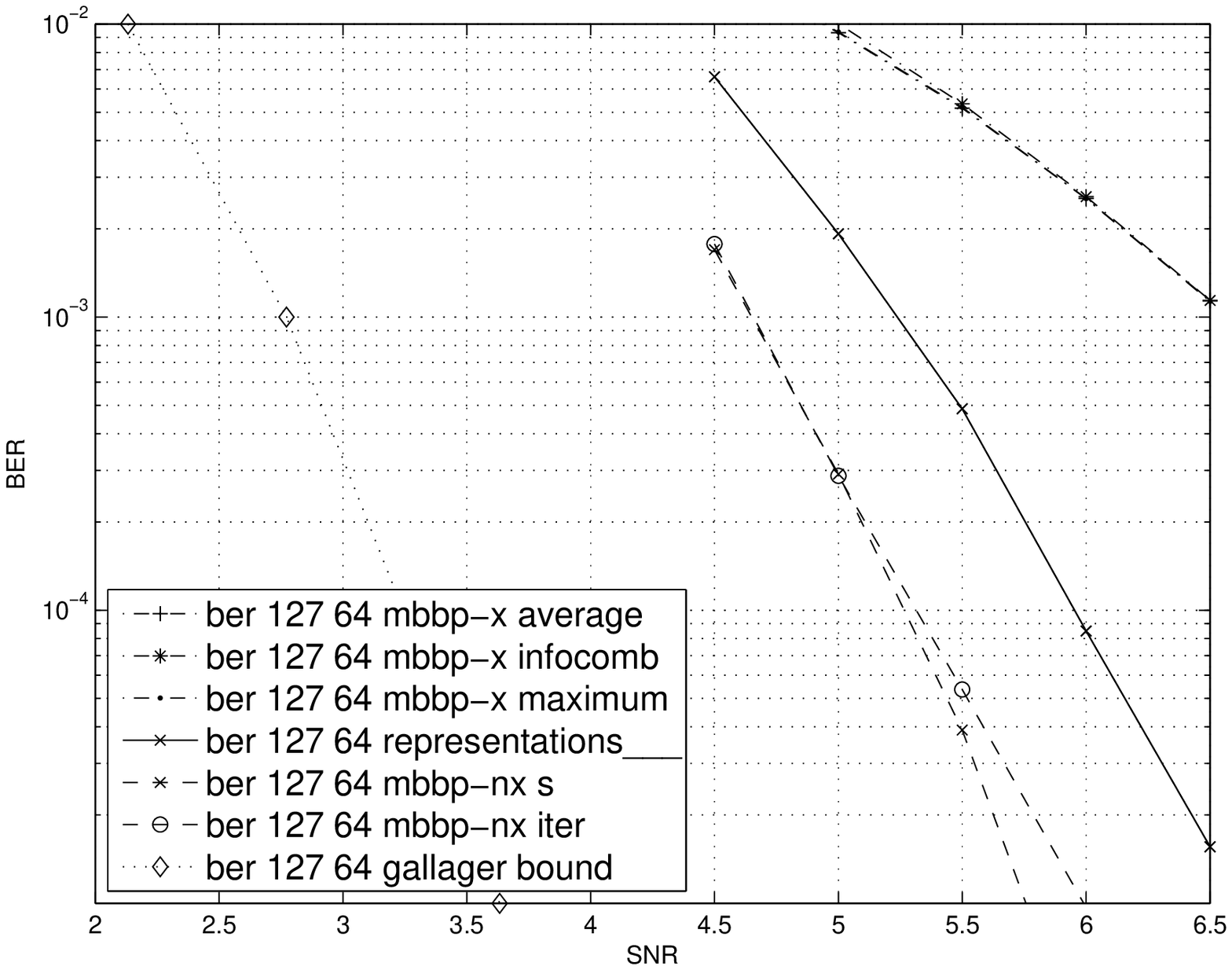}
}\subfigure[Average and maximum number of iterations]{
\psfrag{Average and Max Iter}[cb][cb]{Average $/$ max.\ iterations $\rightarrow$}
\psfrag{SNR}[ct][ct]{$10\log_{10}(\Eb/N_0)\,\rightarrow$}
\psfrag{iter 127 64 representations_______}[l][l]{\scriptsize{Standard BP decoding, avg}}
\psfrag{iter 127 64 mbbp-nx s}[l][l]{\scriptsize{MBBP-NX-S, $l=10$, avg}}
\psfrag{iter 127 64 mbbp-nx iter}[l][l]{\scriptsize{MBBP-NX-FS, $l=10$, avg}}
\psfrag{iter 127 64 representations max}[l][l]{\scriptsize{Standard BP decoding, max}}
\psfrag{iter 127 64 mbbp-nx classic max}[l][l]{\scriptsize{MBBP-NX-S, max}}
\psfrag{iter 127 64 mbbp-nx iter max}[l][l]{\scriptsize{MBBP-NX-FS, max}}
\includegraphics[scale=0.48]{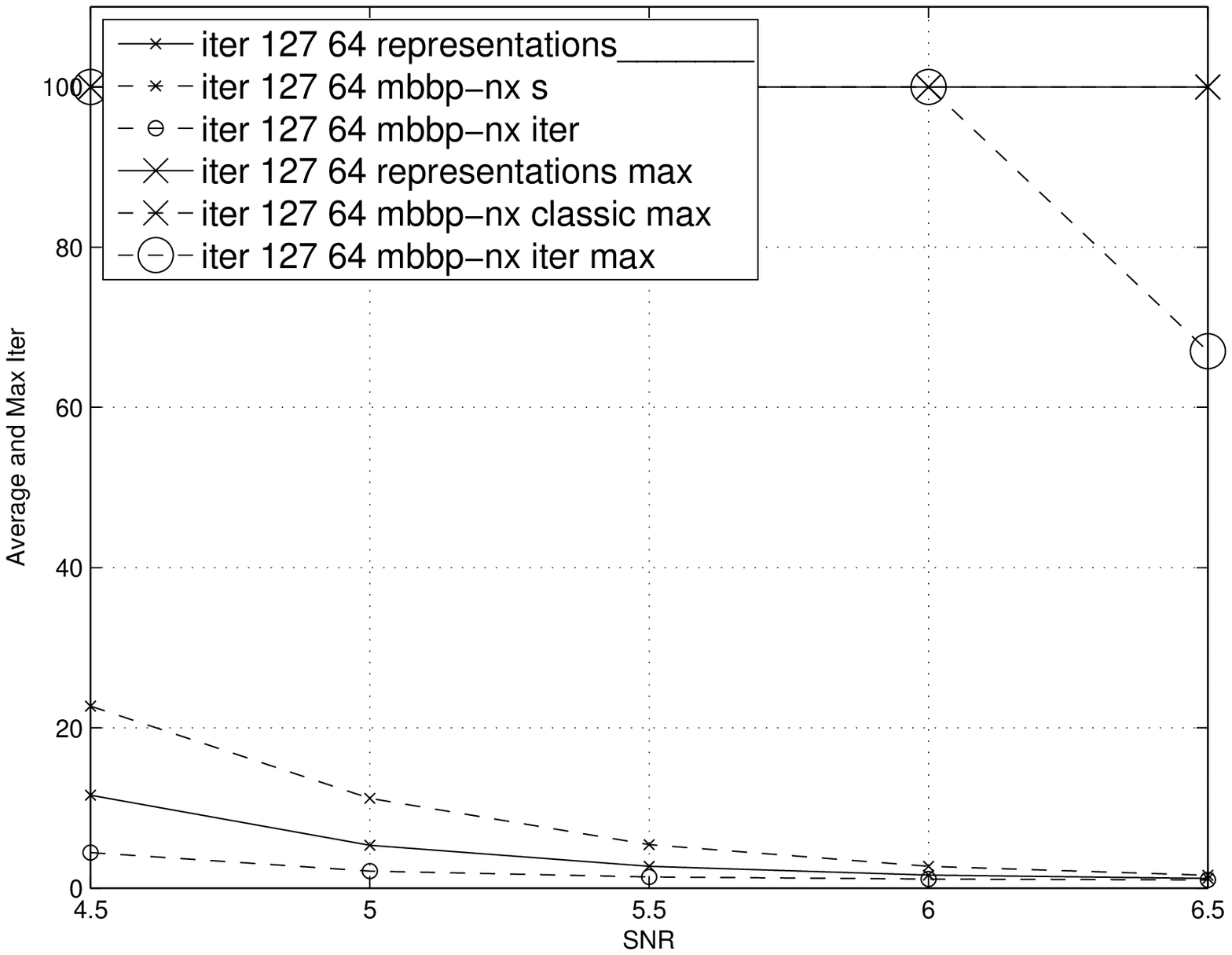}
}
\caption{\label{fig:bch_127_64_oversized_parallel_vs_iterations_mode_10_decoders_100_iter_max_100_iter}Performance comparison and number of iterations for the $[127,64,21]$-BCH code.}
\end{center}
\end{figure}

We observe a behavior that is characteristic for MBBP-NX decoding
in connection with codes of high dimension and co-dimension. A
significant performance improvement compared to BP decoding is
obtained. However, there remains a gap to the Gallager bound. For
this particular code, we observe that MBBP-NX-FS is a very
promising approach. In this case, a significant decrease in the
number of average iterations can be obtained by using the
MBBP-NX-FS decoder, while keeping the decoding performance close
to that of the MBBP-NX-S class. However, all MBBP-X approaches
perform very poorly. Reasons for this phenomena include the fact
that, due to communication between the decoders, the number of
``virtual short cycles'' may be increased.

The performance of the MBBP-NX-S decoder matches the performance obtained with the ordered statistics decoding \cite{fossorieretal95} for first order processing. The approach in
\cite{fossorieretal95} requires for each codeword a sorting of the
received vector according to the received symbol reliability.
Furthermore, a Gauss-Jordan algorithm is applied on the
correspondingly resorted generator matrix, and, in order to obtain a
decoding performance comparable to MBBP-NX-S with $l=10$, up to $64$
pattern tests are performed. We emphasize that MBBP avoids the use of
Gauss-Jordan's algorithm, which is the computationally most costly
component of the decoder in~\cite{fossorieretal95}. Raising the number
of test patterns further, the decoder in~\cite{fossorieretal95} can
outperform the MBBP-NX-S approach with $l=10$ parallel
decoders. Similar observations are made for the Chase Type-2
algorithm. While MBBP-NX-S outperforms Chase Type-2 decoders for codes
of short length \cite{hehnetal07}, Chase Type-2 is superior for longer
codes such as the $[127,64,21]$-BCH code \cite{fossorieretal00}. This
is strongly related to the number of test patterns, which reads
$2^{\lfloor d/2\rfloor}=1024$ and hence significantly higher than the
diversity order in any reasonable MBBP-NX-S setup.

\section{Conclusions}
\label{sec:conclusions}

We introduced a class of decoding algorithms that operate in
parallel on a judiciously chosen family of parity-check matrices.
We considered two variants of this class of techniques: one, in
which the decoders are not allowed to exchange information during
individual runs of the BP algorithm, and another, in which
periodic information exchange is allowed. Algorithms in the first
class were shown to offer significant performance improvements
when compared to the standard BP technique. The approaches in the
second class often compare favorably to standard BP, but do not match
the performance of algorithms that do not make use of periodic
information exchange. Possible reasons for this behavior include
the fact that ``cycles'' are created during the process of
information exchange. These cycles ``in-between'' the graphs of
the representations negatively affect the performance of each
decoder.

The presented approaches were shown to work for classical
high-density codes and are applicable to cyclic and extended cyclic codes. It is possible to generalize the introduced
methods to codes like the progressive edge-growth
(PEG)~\cite{huetal05} family. For this class of codes, significant gains in performance can be obtained \cite{hehn09}.

\bibliographystyle{ieeetr}
\bibliography{LDPC_Group_Bibfile}

\end{document}